\documentclass[draftclsnofoot,onecolumn,12pt]{IEEEtran}
\usepackage{cite,graphics,epsfig,amsmath,amssymb,amsfonts,subfigure,array}
\usepackage{multirow}

\newtheorem{theorem}{Theorem}
\newtheorem{lemma}{Lemma}

\DeclareMathOperator*{\Arg}{arg}

\begin{document}
%
\title{Sensitive White Space Detection with Spectral Covariance Sensing}

\author{{\large Jaeweon~Kim and Jeffrey~G.~Andrews}%
\thanks{The authors are with the 
Department of Electrical and Computer Engineering,
University of Texas at Austin,
Austin, TX 78712, USA
(e-mail: jaeweon@mail.utexas.edu, jandrews@ece.utexas.edu).
This paper was presented in part at IEEE ICASSP 2010, Dallas, TX.  Revision date: \today
}}

\maketitle

\begin{abstract}
This paper proposes a novel, highly effective spectrum sensing algorithm for cognitive radio and whitespace applications.
The proposed spectral covariance sensing (SCS) algorithm exploits the different statistical correlations of the received signal and noise
in the frequency domain. Test statistics are computed from the covariance matrix of a partial spectrogram and compared with a decision threshold
to determine whether a primary signal or arbitrary type is present or not.  This detector is analyzed theoretically and
verified through realistic open-source simulations using actual digital television signals captured in the US. Compared to the state of the art in the literature, SCS improves sensitivity by 3 dB for the same dwell time, which is a very significant improvement for this application.
Further, it is shown that SCS is highly robust to noise uncertainty, whereas many other spectrum sensors are not.

\end{abstract}
%

\section{Introduction}
\label{sec:intro} 

The scarcity of prime (frequencies below 2 or 3 GHz) spectrum is a decades-old problem in wireless communication, and will be for the indefinite future.  Recognizing that a static allocation of spectrum over time and space is highly suboptimal -- for example, often less than 10\% efficient \cite{CabMis04, CorGho07, BroWol04} -- there has been a flurry of interest in finding ways to adaptively allocate spectrum.  Cognitive radio (CR) is a promising approach, whereby transmitter-receiver pairs find unused spectrum (white space) and use it for communication
 \cite{Hoven:thesis,GeiTon07,ZenLia09}.  The idea of spectrum reuse has received regulatory support in the form of the FCC white space ruling authorizing cautious reuse of underutilized spectrum in the licensed TV bands \cite{FCC08-260}.  The IEEE 802.22 standard \cite{802.22/D1.0_2008, SteCho09} is attempting to formalize a solution that will meet FCC approval and allow communication in these bands in the near future.
 
Reliable spectrum sensing -- whereby devices determine whether another (e.g. licensed) signal is present -- is fundamental to the success of cognitive radio.  Primary (licensed) users have priority to use the channel and the secondary users can only use the resource when it is not occupied.  Determining when the channel is occupied is quite difficult, because a potential transmitter must guarantee with high probability that no one in its transmission range could possibly be listening to a primary/licensed signal.  This puts quite strict requirements on the detection capability\footnote{The terms detection and sensing are used
interchangeably throughout this paper.}, and the secondary (i.e. cognitive) system is required to robustly sense the presence of primary signals at very low SNRs. For example, an IEEE 802.22 system must be capable of detecting digital TV signals at an SNR of -21 dB: when the noise is over one hundred times stronger than the actual signal.

\subsection{Spectrum Sensing: Related Work}

Accordingly, spectrum sensing research has been active, resulting in numerous proposed sensing algorithms, which are 
well summarized in \cite{Shellhammer08} and \cite{YucArs09}. Studies on spectrum sensing can be
classified in two categories: blind detection and feature detection. Blind detection is universal since it does not rely on prior information such as signal characteristics, channel and noise power, but the performance is generally poor. Feature detection senses specific characteristics of a {\em known} signal, and typically shows better performance than blind detection at the expense of not being applicable to all possible primary signals.  The simplest (blind) sensing algorithm is an energy detector but it suffers from severe degradation under uncertain noise power \cite{Urkowitz67, SonFis92, TanSah08}. A more robust blind sensing method called covariance absolute value (CAV) detection \cite{ZenLia09} exploits the uncorrelated nature of the noise, whereas the primary signal is correlated. %
While the detection performance is comparable to the energy detector for a given noise power, it is much more robust to uncertain noise power than the energy detector.

There have been several studies on the detection of digital TV signal features, which is one of the primary services to be protected in the IEEE 802.22 systems, and hence a logical starting point for feature detection sensing algorithms. A Field Sync Correlation Detector (FSCD) and a Segment Sync Autocorrelation Detector (SSAD) are proposed in \cite{CheGao07}, which exploit the repetitive nature of synchronization sequences  but their performance is similar to blind detectors. Cordeiro {\it et al.} proposed two FFT based pilot detection methods, one sensing the pilot energy and the other sensing the pilot location \cite{CorGho07}. Pilot location detection is robust to the noise uncertainty while pilot energy detection is not. Under a certain noise power, both the methods achieve fairly good performance, the best in the literature to date as far as the present authors can judge.

The idea of signal detection using cyclic spectral density (CSD) or spectral correlation function (SCF) is introduced in \cite{HanSho06}. The cyclostationary nature of the modulated signal generates unique cyclic frequencies that are known to the detector {\em a priori}. FFT based pilot location detector \cite{CorGho07} is similar to the CSD detector in the sense that they both detect frequencies that contain relatively strong power. 

Tandra and Sahai developed the concept of a useful theoretical limit called the $\mathrm{SNR_{wall}}$ on reliable sensing under uncertain noise power \cite{TanSah07, Tandra:thesis}. They also proposed a matched filter with run-time noise calibration \cite{TanSah08} to combat the time-varying nature of the wireless channel. Although the matched filter is an optimal detection algorithm, it requires perfect synchronization and parameter estimation, which is generally impractical.

\subsection{Contributions and Organization}

In this paper, we propose a novel spectral covariance sensing (SCS) algorithm that exploits different statistical correlations of the signal and noise in the frequency domain.  
The SCS detects spectral features for maximum sensitivity, but also is applicable for non-flat spectrum signals. We analyze its performance theoretically and then verify those results through extensive simulations using actual DTV signals captured in a real environment.  Our simulations are publicly available, so that our performance claims can be externally verified \cite{SCS-CODE-URL}, which seems particularly important for this application \cite{VanVet09}.

Rigorous comparisons with the FFT based pilot energy detector \cite{CorGho07} and the CAV detector \cite{ZenLia09} show that SCS achieves 3 dB better sensitivity than the FFT based pilot detection method (and an even larger gain vs. CAV) with the same sensing time or equivalently, achieves the same sensitivity in 20\% of the sensing time. In the extremely low SNR regime, a 3 dB gain is a very significant improvement.  It is also shown that the sensitivity of the CAV can be improved by a lot when combined with parts of the proposed algorithm.  Furthermore, we show that the technique is unusually robust to noise uncertainty. The SCS achieves the best sensitivity of any spectrum sensing approach thus far proposed in the literature, to the best of our knowledge.

The rest of this paper is organized as follows. The SCS algorithm is presented in Section \ref{sec:SensingAlgorithm}. Section \ref{sec:SCSAnalysis} analyzes detection performance of the SCS, especially in the low SNR regime. The theoretical results are verified through simulations using actual data and requirements of the IEEE 802.22 standards in Section \ref{sec:SimResults}. Detection performance with various parameter selection is explored and comparisons with the previous methods are also presented. 
Section \ref{sec:conclusion} concludes the paper.

\section{Spectral Covariance Sensing}
\label{sec:SensingAlgorithm}

Assume that there is either one or no primary transmitter to detect, and
that the secondary node can be located inside or outside of the primary cell boundary.
The detection problem can then be formulated as a binary decision under two hypotheses as in \cite{TanSah08,ZenLia09}:
\begin{eqnarray}
    \mathcal{H}_0: z(n) &=& w(n), \\
    \mathcal{H}_1: z(n) &=& s(n) + w(n),
\end{eqnarray}
where $z(n)$ is the equivalent received signal at baseband, $s(n)$ is the signal component of received samples and $w(n)$ is the noise component. The signal component $s(n)$ may have a DC term ($s_d(n) = s_d$) and an AC term ($s_a(n)$). Therefore the signal samples are assumed to be independent and identically-distributed (i.i.d) with mean $\mathbb{E}[s(n)] = s_d$ and variance $\sigma_s^2 = \mathbb{E}[s_a(n)^2]$.
The noise samples $w(n)$ are zero-mean i.i.d white Gaussian noise with variance $\sigma_w^2$.

For simplicity, we assume that the primary signal has a flat spectrum with a pilot tone at low frequency, consistent with the ATSC signal \cite{CheGao07, CorGho07}. However, the proposed algorithm can also be applied to signals with {\em non-flat} spectrum.
The SCS algorithm exploits correlation of spectral feature in the band of interest 
where the signal specific spectrum that distinguishes itself from other signal or noise is located. 
Once the spectrum of the received signal is obtained by periodogram estimation \cite{ProakisBook96},
its correlation is computed using the sample covariance matrix.
The primary signal typically has unique non-flat spectrum, so it is highly correlated, whereas the noise spectrum is flat and is almost entirely uncorrelated.
To find the relative correlation, the autocovariance of the spectrum is compared with the total covariance.
The signal is detected if the spectral correlation is higher than the predefined threshold.
The proposed SCS algorithm is described as the following steps.

    \begin{enumerate}
      \item Downconvert the received signal $x(t)$ to a baseband complex signal $y(t) = x(t)e^{-j2\pi f_c t}$
          to place the pilot tone at or near DC, where $f_c$ is the pilot frequency.
      \item Low pass filter and downsample $y(t)$ by appropriate sampling rate $F_s$ 
      to form
      $z(n)$. 
      \item Compute $z(n)$'s spectrogram by calculating the squared magnitude of its short-time Fourier transform (STFT) as
      \begin{equation}
        Z_\tau(k) = \frac{1}{N}\left|\sum_{n=0}^{N-1} z(n + \tau N)e^{-j\frac{2\pi nk}{N}}\right|^2,
      \end{equation}
         where $N$ is the number of FFT points,
	 $\tau \in \left\{0, 1, \cdots, N_d-1\right\}$ is the index of the sensing window, 
	 $N_d$ is the total number of sensing windows
	 and $k \in \{-N/2, \cdots, 0, \cdots, N/2-1\}$ is the frequency index.
         Hence, the FFT is calculated at every dwell time ($t_s$) for $N_d$ times. 
      \item Select components near DC according to 
          \begin{eqnarray}
              \mathbf{M} &=& \left[\begin{array}{cccc}
                  Z_0(-K) & Z_1(-K) & \cdots & Z_{N_d-1}(-K)\\
                  Z_0(-K+1) & Z_1(-K+1) & \cdots & Z_{N_d-1}(-K+1)\\
                  \vdots & \vdots & \vdots & \vdots\\
                  Z_0(0) & Z_1(0) & \cdots & Z_{N_d-1}(0)\\
                  \vdots & \vdots & \ddots & \vdots\\
                  Z_0(K) & Z_1(K) & \cdots & Z_{N_d-1}(K)
                  \end{array} \right]\\
                 &\triangleq& \big[\mathbf{m_0} \qquad \mathbf{m_1} \qquad  \cdots \qquad  \mathbf{m_{N_d-1}} \big]
          \end{eqnarray}
          where $K$ is the index of low pass filter cut off
          frequency ($B_f$) in the FFT, i.e. $K = \lfloor N \cdot B_f/F_s \rfloor$.
          This matrix reduction is in effect a low pass filter, that selects spectral feature of the primary signal and reduces noise power.
	  The effective SNR is increased accordingly as long as the selected signal portion contains more power than the rest does in the same bandwidth.
	  Algorithmic complexity is also reduced by a factor of $N/(2K+1)$ due to the smaller matrix size in the covariance calculation.
      \item Calculate sample covariance of $\mathbf{M}$ as
        \begin{eqnarray}
            \mathbf{C} &=& \mathrm{cov}(\mathbf{M}) \nonumber \\
	    &=& \mathbb{E}_k\left[(\mathbf{M} - \boldsymbol{1_{N_d}} \, \overline{\mathbf{M}})^T (\mathbf{M} - \boldsymbol{1_{N_d}} \, \overline{\mathbf{M}}) \right]
        \end{eqnarray}
        where 
          $\overline{\mathbf{M}} = \boldsymbol{\mu_M} = [\mu_0, \mu_1, \cdots, \mu_{N_d-1}]$, $\mu_\tau = \overline{\mathbf{m_\tau}} = \frac{1}{2K+1} \sum_k m_\tau(k)$, 
	  $\boldsymbol{1_{N_d}} = [1, 1, \cdots, 1]^T$ is the all one's vector with length $N_d$ and $\mathbb{E}_k[\cdot] = \frac{1}{2K} \sum_k [\cdot]$.
	  Note that the covariance matrix $\mathbf{C}$ is symmetric.
      \item Compute test statistic $T = T_1/T_2$, where
          \begin{eqnarray}
              T_1 &=& \frac{1}{N_d} \sum_{\tau=0}^{N_d-1} \sum_{u=0}^{N_d-1} c_{\tau u} = T_2 + \frac{2}{N_d} \sum_{\tau=0}^{N_d-1} \sum_{u=\tau+1}^{N_d-1} c_{\tau u}\\
              T_2 &=& \frac{1}{N_d} \sum_{\tau=0}^{N_d-1} c_{\tau \tau},
          \end{eqnarray}
          where $c_{\tau u}$ is the element of $\mathbf{C}$ at the $\tau$-th row and $u$-th column, which is the covariance of $\mathbf{m_\tau}$ and $\mathbf{m_u}$. In other words, $T_2$ is the sample mean of the diagonal terms of the $\mathbf{C}$ matrix, which are the autocovariances of the spectrograms, and $T_1$ is the sum of all the elements of the covariance matrix.
      \item Compare $T$ with decision threshold $\gamma$ as
      \begin{equation}
	T  \overset{\mathcal{H}_1}{\underset{\mathcal{H}_0}{\gtrless}} \gamma, \,\,
	\mathrm{s.t.} \,\,\, \gamma  = \Arg_\mathcal{T} \sup P_{FA}(\mathcal{T}) = P_{FA,req}, \label{eq:gamma_def}
      \end{equation}
      where $P_{FA}(\mathcal{T})$ is the probability of false alarm with threshold $\mathcal{T}$
      and $P_{FA,req}$ is the required false alarm probability given by the specification. If $T$
      exceeds $\gamma$, detection is declared $\left(\mathcal{H}_1 \right)$. Otherwise, the decision is made that no signal is present $\left(\mathcal{H}_0 \right)$.
    \end{enumerate}

The test statistics $T$ may superficially look similar to \cite{ZenLia09}, but there are several key differences that lead to SCS's better performance. 
First, \cite{ZenLia09} used correlation in the time domain while SCS uses the frequency domain.
As shown in \cite{CorGho07} and  \cite{HanSho06}, power spectrum is one of the most significant features of the primary signal to be detected and is a frequency domain measure. Note that SCS detects unique spectral features that are not entirely flat, e.g. peaks or valleys, while \cite{ZenLia09} can be applied to any modulated signal without prior information.
Second, SCS uses only part of the spectrum by low pass filtering in step 4), which results in reduced noise power, in other words, increased effective SNR. Downsampling in step 2) also helps by averaging out noise which has zero mean. Low pass filtering in step 2) is essential to prevent aliasing and suppresses noise power in stop band. Engaging this step can improve the detection performance of \cite{ZenLia09} considerably.
Third, \cite{ZenLia09} uses non-negative (i.e. absolute) values of the covariance matrix but SCS uses both positive and negative values. 

\section{Performance Analysis}
\label{sec:SCSAnalysis}

The performance of spectrum sensing can be measured by two probabilities: probability of detection $P_D$ and probability of false alarm $P_{FA}$, which are defined as
\begin{eqnarray}
    P_D &=& \Pr \left(T > \gamma \bigm| \mathcal{H}_1 \right), \label{eq:PD_def}\\
    P_{FA} &=& \Pr \left(T > \gamma \bigm| \mathcal{H}_0 \right).  \label{eq:PFA_def}
\end{eqnarray}
The primary goal of the proposed sensing algorithm is achieving high $P_D$ and low $P_{FA}$ at the lowest received signal power level,
however these two probabilities trade off each other depending on the decision threshold $\gamma$.
Therefore the required detection capability is determined by the application as a minimum $P_D$ and maximally allowed $P_{FA}$ pair at the required SNR.

In this section, detection performance of the SCS algorithm is analyzed according to the following steps.
Note that the current signal model is one of the most difficult cases for the SCS detector, since it has only one peak over a wide flat spectrum. The analysis given in this section can be easily applied to other signal models with other spectral features, by suitably adjusting the value(s) of $f_c,\, K, \,B_f, \,t_s,$ and $N_d$.
We will first set up the analytical model and derive the statistics of variables.
Second, we will find a solution for $P_{FA}$. Since the noise is assumed to be white Gaussian and the $P_{FA}$ analysis needs only the noise spectrum and its correlation, we can get the analytical results accurately.
Then, the decision threshold $\gamma$ is obtained by setting the $P_{FA}$ as required.
Finally, the $P_D$ is analysed with the given threshold and received SNR.
In order to find the theoretical solution for $P_D$, the statistical distribution of $T = T_1/T_2$ is needed, however it is almost impossible to get in practice, especially when the primary signal can have arbitrary spectrum shape as allowed in the SCS algorithm. Therefore, we assume Gaussian distribution thanks to the central limit theorem as $P_{FA}$ case.
Besides the Gaussian assumption, we simplifies some of the statistics for very low SNR regime that enables to find a closed form solution.
A summary of the notations used in the analysis is given in Table \ref{tbl:Notations}.
The convention used in this paper is that the subscript $p$, $s$ and $w$ stands for pilot, the rest of the signal and noise respectively.


\subsection{Analytical Model and Statistics}
\label{subsec:stats}

The partial matrix $\mathbf{M}$ captures spectrum on the band of interest, 
such as the spectrum near the pilot for the ATSC signal.
Ideally, white noise would have a flat spectrum over the frequency band; however this cannot be true in a practical system using periodogram estimation.
Hence the noise spectrum can be estimated by some fluctuations over the nominal noise power spectrum density (PSD) $N_w/2$. 
The residual components, i.e. the fluctuations over the nominal spectrum, are weakly correlated due to the filtering, however the amount is small compared to the signal correlation.
Since the diagonal elements of the covariance matrix $\mathbf{C}$ are the autocovariances, they are larger than the off-diagonal ones.
When there is no signal components ($\mathcal{H}_0$), $\mathbf{C}$ becomes almost a diagonal matrix and the resulting test statistic $T$ gets close to 1.
When there is a primary signal whose spectrum is not flat, the residual spectrum is highly correlated and $T$ increases.

When there is no signal, the spectrum is estimated as
\begin{eqnarray}
    m_\tau(k) &=& \frac{1}{N}\left|\sum_{n=0}^{N-1} w_\tau(n)e^{-j\frac{2\pi kn}{N}}\right|^2 \\
              &=& N_w(\tau) + \delta_w(\tau, k),
\end{eqnarray}
and when the signal is present, it is
\begin{eqnarray}
    m_\tau(k) &=& \frac{1}{N}\left|\sum_{n=0}^{N-1} (s_\tau(n) + w_\tau(n) )e^{-j\frac{2\pi kn}{N}}\right|^2 \nonumber \\
              &=& \frac{1}{N}\sum_n \sum_m \left\{s_\tau(n)s_\tau^\ast(m) e^{-j\frac{2\pi k(n-m)}{N}} +
                      w_\tau(n)w_\tau^\ast(m) e^{-j\frac{2\pi k(n-m)}{N}} \right. \nonumber \\
              & &  \;\;\;\; \left. + 2 \mathrm{Re}\left\{s_\tau(n)w_\tau^\ast(m) e^{-j\frac{2\pi k(n-m)}{N}} \right\} \right\} \nonumber \\
	      &=&  N_s(\tau) + \delta_s(\tau, k) + N_w(\tau) + \delta_w(\tau, k) + \frac{2}{N}\mathrm{Re}\left\{S_\tau(k) W_\tau^\ast(k)\right\}. \label{eq:m_tau_signal}
\end{eqnarray}
The signal and noise components are assumed to be uncorrelated.

If there is no signal, the statistics of $m_\tau(k)$ are given as in \cite{ProakisBook96}
\begin{eqnarray}
  \mathbb{E} \left[m_\tau(k)\right] &=& N_o/2 = N_w/2 \triangleq \mu_\tau, \\
    \mathrm{var}\left(m_\tau(k)\right) &=& \left(N_w/2\right)^2, 
\end{eqnarray}
where $N_o$ is the one-sided PSD of the noise signal $w(t)$, which has full span of the bandwidth $B$. Note that $\mu_\tau \neq \mathbb{E}[\mathbf{m}_\tau(k)]$ for non-white signal. Since $w(n)$ is the filtered and downsampled version of $w(t)$, the effective noise power is given as
\begin{equation}
    \sigma_w^2 = \sigma_o^2 B_f/B = N_w B_f,
\end{equation}
where $\sigma_o^2$ is the noise power in received signal $x(t) = w(t)$. A perfect low pass filter with cutoff frequency $B_f$ is assumed here.

In this paper, we assume
%
the primary signal has a pilot tone, 
which has $f_p (0 < f_p < 1)$ of the total signal power $P_S$.
With these settings, we can obtain the mean PSD of the received signal as follows.

\begin{lemma}
  \label{lem:mu_tau}
The in-band average spectrum of the signal and noise at the $\tau$-th sensing window, which is the mean of the $\tau$-th column of the matrix $\mathbf{M}$, is
  \begin{equation}
    \mu_\tau  = N_w + \frac{\delta_p(\tau)}{2K} + \frac{N_s(\tau)}{2}.
  \end{equation}
\end{lemma}
\begin{proof}
  See Appendix \ref{apdx:proof_mu_tau}.
\end{proof}

In the low SNR regime where $N_w \gg N_s$, $\mu_\tau$ can be further simplified as
\begin{equation}
    \mu_\tau  \approx  N_w + \frac{\delta_p(\tau)}{2K}, \label{eq:Nw+Ns}
\end{equation}
i.e. the in-band power is mostly concentrated on the pilot frequency and it can be comparable to the noise power even in the extremely low SNR where the rest of the signal power is almost dominated by the noise components.

The elements of the covariance matrix $\mathbf{C}$ can be obtained as
\begin{eqnarray}
  c_{\tau u} &=& c_{u \tau} = \mathbb{E}_k\left[(\mathbf{m_\tau} - \mu_\tau\boldsymbol{1_{2K+1}})^T(\mathbf{m_u} -
  \mu_u\boldsymbol{1_{2K+1}})\right] \\
        &=& \frac{1}{2K} \sum_{k=-K}^K \left(m_\tau\left(k\right) - \mu_\tau\right) \left(m_u\left(k\right) - \mu_u \right) \\
	&=& \mathbb{E}_k[\mathbf{m_\tau}^T \mathbf{m_u}] - \mu_\tau \mu_u \left(\frac{2K+1}{2K}\right). \label{eq:E[c_tau_u]}
\end{eqnarray}

Define normalized correlations for signal power spectrum and noise as
\begin{eqnarray}
  \alpha_p^l &=& \frac{\mathbb{E}[\delta_p(\tau) \delta_p(u)]}{\delta_p^2} > 
  	\frac{\mathbb{E}[N_s(\tau) N_s(u)]}{N_s^2} = \alpha_s^l, \\
  \alpha_w^l &=& \frac{\mathbb{E}_k[(\mathbf{d}_w^\tau)^T \mathbf{d}_w^u]}{\sqrt{\mathrm{var}_k(\mathbf{d}_w^\tau) \mathrm{var}_k(\mathbf{d}_w^u)}}
  = \frac{\frac{1}{2K} \sum_k \delta_w(\tau, k) \delta_w(u, k)}{N_w^2},
\end{eqnarray}
where $l = \tau - u$ and $\mathbf{d}_w^\tau = [\delta_w(\tau,-K), \cdots, \delta_w(\tau,K)]^T$.

Then the statistics of $c_{\tau u}$ can be obtained using the following Lemma \ref{lem:Ectu}, which specifies the covariance matrices. 

\begin{lemma}
  \label{lem:Ectu}
  When the primary signal is present, the statistics of $c_{\tau u}$ can be derived as
  \begin{eqnarray}
    \mathbb{E}[c_{ \tau u }] 
      &=& \alpha_w^l N_w^2 + \frac{\alpha_p^l}{2K}\delta_p^2 + N_s \left(N_w + (N_s + N_w)\alpha_s^l \right) + O(\frac{1}{K^2}), \label{eq:Ectu_Sig} \\
    \mathrm{var}[c_{ \tau u }] &=&  \frac{N_w^2\delta_p^2}{2K^2} \left(1 + \frac{\delta_p}{N_w}\right)^2 \left(1 + \alpha_p^l \delta(l) \right) +
    \frac{N_w^2}{4K^2} \left( K + \left(1+\frac{\delta_p}{N_w}\right)^2 \right) \left( 2 + N_w^2\delta(l) \right) \nonumber \\
     & & \;\;\; + O(\frac{1}{K^4}),  \label{eq:Vctu_sig}
  \end{eqnarray}
  where $l = \tau - u$ and
	$\delta(n)$ is a Kronecker delta function, which is $\delta(n) = \left\{\begin{smallmatrix} 1, & \text{if $n = 0$} \\
	  			0, & \text{if $n \neq 0$} \end{smallmatrix} \right.$.
\end{lemma}

\begin{proof}
  See Appendix \ref{apdx:proof_Ectu}.
\end{proof}

Since $K \gg 1$, the last terms of (\ref{eq:Ectu_Sig}) and (\ref{eq:Vctu_sig}) can be omitted for simplification. As in (\ref{eq:Nw+Ns}), the expectation of $c_{\tau u}$ can be further simplified in the low SNR regime as
\begin{equation}
    \mathbb{E}[c_{ \tau u }] \approx \alpha_w^l N_w^2 + \frac{\alpha_p^l}{2K} \delta_p^2. \label{eq:Ectu_sig}
\end{equation}
Note that the no signal case $\mathcal{H}_0$ can be found by letting $N_s = \delta_p = 0$.

The Lemma \ref{lem:Ectu} clearly shows that the covariance matrix depends on the correlation of the received signal and its relative power to the noise power, which is the SNR.
%
%
%
The statistics of the test are given in the following Theorem \ref{thm:ET1T2}. 

\begin{theorem}
  \label{thm:ET1T2}
  The statstics of the test crieterion in the low SNR regime are
  \begin{eqnarray}
    \mathbb{E}[T_1(N_d)] &\approx& N_w^2 + \frac{\delta_p^2}{2K}
      + \frac{2A_{N_d}^w}{N_d}N_w^2 + \frac{A_{N_d}^p}{KN_d}\delta_p^2 , \label{eq:ET1_sig} \\
      \mathbb{E}[T_2(N_d)] &\approx& N_w^2 + \frac{\delta_p^2}{2K}, \label{eq:ET2_Sig} \\
    \mathrm{var}(T_2(N_d)) &\approx& \frac{1}{K^2N_d} \left\{\delta_p^2 (N_w + \delta_p)^2 + \frac{N_w^2}{4}\left(N_w^2K + \left(N_w + \delta_p\right)^2 \right) \right\},  \label{eq:VarT2_Sig}
  \end{eqnarray}
  where the accumulated correlation of the noise and signal are given as
  \begin{eqnarray}
    A_{N_d}^w &=& \sum_{l = 0}^{N_d-2} (N_d - 1 - l) \alpha_w^l, \\
    A_{N_d}^p &=& \sum_{l = 0}^{N_d-2} (N_d - 1 - l) \alpha_p^l.
    \label{eq:A_Nd}
  \end{eqnarray}
\end{theorem}
Again the no signal case is obtained by letting all the signal terms to be zero and the approximations become exact.

\begin{proof}
Equations (\ref{eq:ET1_sig}) through (\ref{eq:VarT2_Sig}) can be easily derived from the results of Lemma \ref{lem:Ectu} and (\ref{eq:Ectu_sig}) as
  \begin{eqnarray}
    \mathbb{E}[T_2] &=& \mathbb{E}[c_{\tau \tau}],\\
    \mathbb{E}[T_1] &=& \mathbb{E}[T_2] + \frac{2}{N_d} \sum_{\tau = 0}^{N_d-2} \sum_{u = \tau+1}^{N_d-1} \mathbb{E}[c_{\tau u}] \\
    \mathrm{var}[T_2]&=& \frac{1}{N_d} \mathrm{var}(c_{\tau \tau}).
    \label{eq:pf_thmET1T2}
  \end{eqnarray}
\end{proof}

The approximations are made here based on the simplified statistics of the covariance matrix for the low SNR regime, and will be verified through simulation in Section \ref{sec:SimResults}.
The results of the Theorem \ref{thm:ET1T2} shows that $T_2$ is composed of autocovariances of the signal and noise term and $T_1$ is composed of accumulated covariances. Therefore, we can expect the relative amount of these values will determine the detection performance,
which will be derived in Section \ref{subsec:ProDet}.

\subsection{Probability of Detection}
\label{subsec:ProDet}

In this section, the performance of the SCS algorithm is analyzed in two steps. The decision threshold is found first and then the probability of detection is derived as in \cite{ZenLia09}. By the central limit theorem, both $T_1$ and $T_2$ approach Gaussian distribution when the number of samples ($N_d \times N$) becomes large.
The decision threshold and the corresponding probability of detection are given in the following Theorem \ref{thm:gamma} and \ref{thm:P_D}.

\begin{theorem}
  \label{thm:gamma}
  The decision threshold $\gamma$ is determined as
\begin{equation}
  \gamma = \frac{2\sqrt{KN_d} \left( 1 + \frac{2A_{N_d}^w}{N_d}\right)}{Q^{-1}(1 - P_{FA,req}) + 2\sqrt{KN_d}}, \label{eq:gamma}
\end{equation}
where $P_{FA,req}$ is the required probability of false alarm.
\end{theorem}

\begin{proof}
The probability of false alarm $P_{FA}$ can be calculated as 
\begin{eqnarray}
  P_{FA} &=& \mathbb{P} \left(\frac{T_1}{T_2} > \gamma \Bigm| \mathcal{H}_0 \right) \nonumber \\
    &\approx& \mathbb{P} \left(\frac{T_2 - \mathbb{E}[T_2]}{\sqrt{\mathrm{var}(T_2)}} < \frac{\mathbb{E}[T_1] / \gamma - \mathbb{E}[T_2]}{\sqrt{\mathrm{var}(T_2)}}\Bigm| \mathcal{H}_0 \right)\nonumber \\
    &=& \mathbb{P} \left(\frac{T_2 - N_w^2}{\sqrt{\frac{N_w^4}{4KN_d}}} < \frac{\frac{N_w^2}{\gamma}\left(1 + \frac{2A_{N_d}^w}{N_d}\right) - N_w^2}{\sqrt{\frac{N_w^4}{4KN_d}}} \Bigm| \mathcal{H}_0 \right)\\
    &=& 1 - Q\left( \frac{2\sqrt{KN_d}}{\gamma}\left(1 + \frac{2A_{N_d}^w}{N_d}\right) - 2\sqrt{KN_d} \right) \label{eq:P_FA},
\end{eqnarray}
where \[Q(x) = \frac{1}{2\pi} \int_x^\infty e^{-x^2/2} dx.\]

Therefore, $\gamma$ is easily derived from (\ref{eq:gamma_def}) and (\ref{eq:P_FA}).
\end{proof}

If we assume the noise correlation is independent of the lag, i.e. $\alpha_w^l = \alpha_w, \forall l$, then $A_{N_d}^w$ becomes proportional to $N_d^2$ for large values of $K \cdot N_d$ product and the threshold $\gamma$ has a tendency to become proportional to $N_d$. Note that the decision threshold is independent of the noise power and SNR, which explains the robustness of SCS to the noise power uncertainty.

\begin{theorem}
  \label{thm:P_D}
The probability of detection is obtained as
\begin{eqnarray}
  P_D &=& \mathbb{P} \left(\frac{T_1}{T_2} > \gamma \Bigm| \mathcal{H}_1 \right) \label{eq:P_Danal} \\
    &\approx& 1 - Q\left( \frac{ \frac{2\sqrt{N_d}}{\gamma} \left(1 + \frac{2A_{N_d}^w}{N_d} - \gamma \right) + \frac{N_d+2A_{N_d}^p - \gamma N_d}{\gamma K \sqrt{N_d}} \Gamma_p^2 }{1 + \frac{2\Gamma_p^2}{K}} \right),
\end{eqnarray}
where $\Gamma_p$ is the effective SNR, which is
\begin{equation}
  \Gamma_p \triangleq \frac{\delta_p}{N_w} = \frac{f_pBN}{F_s}\frac{P_S}{P_N} = \frac{f_p BN}{F_s} \mathrm{SNR}. \label{eq:Gamma_p} \\
\end{equation}

\end{theorem}

\begin{proof}
The probability of detection can be derived in a similar way as Theorem \ref{thm:gamma}.
The variance of $T_2$ can be further approximated for the low SNR case as
\begin{eqnarray}
  \mathrm{var}(T_2(N_d)) 
    &\leqq& \frac{1}{N_d}\left(\frac{\delta_p}{K}\left(N_w+\delta_p\right) + \frac{N_w}{2K}\left(N_wK + \frac{\delta_p}{K} \right) \right)^2 \\
    &=& \frac{1}{N_d}\left( \frac{N_w^2}{2} + \frac{\delta_p}{K} N_w + \frac{\delta_p^2}{K} \right)^2 \\
    &\approx& \frac{1}{4N_d} \left(N_w^2 + \frac{2\delta_p^2}{K} \right)^2.  \label{eq:VarT2_Sig_Apprx}
\end{eqnarray}

Then, the probability of detection is derived the same way as the proof of Theorem \ref{thm:gamma} as
\begin{eqnarray}
  P_D &=& \mathbb{P} \left(\frac{T_1}{T_2} > \gamma \Bigm| \mathcal{H}_1 \right) \nonumber \\
  &=& \mathbb{P} \left(\frac{T_2 - \mathbb{E}[T_2]}{\sqrt{\mathrm{var}(T_2)}} < \frac{\mathbb{E}[T_1] / \gamma - \mathbb{E}[T_2]}{\sqrt{\mathrm{var}(T_2)}} \Bigm| \mathcal{H}_1 \right)\nonumber \\
  &=& \mathbb{P} \left( \frac{T_2 - \mathbb{E}[T_2]}{\sqrt{(\mathrm{var}(T_2))}} <
    \frac{ \frac{N_w^2}{\gamma}\left(1+\frac{2A_{N_d}^w}{N_d}\right) + \frac{N_d+2A_{N_d}^p}{2\gamma KN_d}\delta_p^2 - N_w^2 - \frac{\delta_p^2}{2K}}{\frac{1}{2\sqrt{N_d}} \left(N_w^2 + \frac{2\delta_p^2}{K} \right)} \Bigm| \mathcal{H}_1 \right) \\
    &\approx& 1 - Q\left( \frac{ \frac{2\sqrt{N_d}}{\gamma} \left(1 + \frac{2A_{N_d}^w}{N_d} - \gamma \right) + \frac{N_d+2A_{N_d}^p - \gamma N_d}{\gamma K \sqrt{N_d}} \Gamma_p^2 }{1 + \frac{2\Gamma_p^2}{K}} \right).
	\label{eq:P_DAnal}
\end{eqnarray}
\end{proof}

Given SNR and sampling rate $F_s$, the effective SNR is determined by 
the number of FFT $N$ in a dwell, since the signal bandwidth $B$ and pilot power fraction $f_p$ are fixed for the ATSC signal.
Theorem \ref{thm:P_D} indicates
that the detection probability $P_D$ can be increased by increasing the number of samples either by $N$ or $N_d$. Note that $K$ is also increased when $N$ is increased. Also note that the decision threshold $\gamma$ is increased along with $N_d$.
Hence, we expect that the length of sensing window or dwell $t_s$ affects the performance more than the number of dwells $N_d$.

\section{Simulation Results and Comparisons} 
\label{sec:SimResults}


In this section, we present simulation results of the proposed SCS algorithm that show its practicality and superior detection performance versus existing solutions in terms of sensitivity and sensing time. 
The simulation has been conducted in compliance with the IEEE 802.22 standards using actual DTV signals captured in the US, for fair comparison and more importantly, to prove SCS's applicability to real spectrum sensing problems.
Our simulation codes are publicly available in \cite{SCS-CODE-URL}.
The analytical results driven in Section \ref{sec:SCSAnalysis} are also verified in this section.
\subsection{Spectrum Sensing Requirements}
\label{subsec:requirements}

In any CR based system, emphasis is on the sensing of potential primary services in the given frequency band.
Each CR system has its own requirements on the spectrum sensing depending on various aspects of the primary services, propagation characteristics and
application specific features.
Spectrum sensing requirements of the IEEE 802.22 system are reviewed in this section and 
used in the simulation.

The sensitivity is defined as the minimum received power at the secondary system that can be detected
with the given probability of detection ($P_D$) and  probability of false alarm ($P_{FA}$) pair. Note
that the sensitivity is not equal to the $\mathrm{SNR_{wall}}$ in \cite{TanSah08}. The primary services in
the IEEE 802.22 systems are digital TV, analog TV, and wireless microphone, 
operated in TV Bands
(54 -- 862 MHz) \cite{FCC08-260, 802.22_Req}. 
Note that the format and bandwidth of the TV signal depends on the
regional regulation. For example, digital TV signals can be Advanced Television Systems Committee (ATSC: North America),
Digital Video Broadcasting - Terrestrial (DVB-T: Europe)
or Integrated Services Digital Broadcasting (ISDB: Japan) \cite{SteCho09, CorGho07}. In this paper, only ATSC A/74 DTV signal
\cite{ATSC:A74} is considered for detection.
The required sensitivity is 90\% of $P_D$ and 10\% of $P_{FA}$ at -116 dBm, which is equivalent to -21 dB of SNR. 
Figure \ref{fig:ATSCspectrum} shows a sample ATSC spectrum that will be detected using the proposed algorithm.
As explained in Section \ref{sec:SCSAnalysis}, it has a high powered pilot on the lower edge and is correlated near the pilot.

\subsection{Simulation Procedure}
\label{subsec:SimProcedure}

The simulations were performed according to the procedures in \cite{802.22_Sim_model} and
\cite{802.22:InitialSigProc} with the proposed SCS algorithm. Two previous sensing algorithms are
also tested for comparison and the performance analysis given in Section \ref{sec:SCSAnalysis} is also verified.
The actual DTV signals captured by the ATSC in
Washington DC and New York City are used for simulation, which includes all relevant impairments of
wireless channels such as fading, scattering, shadowing and local oscillator mismatch. The
12 selected signals recommended by the 802.22 committee \cite{802.22_Sim_model} are used as in several other studies \cite{CorGho07, ZenLia09, CheGao07}.

Due to the random
characteristics of the wireless channel, detector performance varies significantly.
Figure \ref{fig:SCS_Nd30_ts1_rho0} shows the probability of missed detection ($P_{MD} = 1 - P_D$) for five selected signals that illustrates the broad range of sensing performance along with the average of the 12 signals.
Note that the sensitivity difference
can be as large as 12 dB with $P_{FA} = 0.1$. The 21.52 MHz sampled ATSC captured signal is
downsampled to $F_s$ = 2.152 MHz. The FFT point $N$ is determined by
\begin{equation}
    N = 2^n \; s.t. \; n = \lfloor \log_2 (F_s/\Delta f) \rfloor = \lfloor \log_2 (F_s \cdot t_s)
    \rfloor, \label{eq:N}
\end{equation}
where $\Delta f$ is the frequency resolution of the FFT and $t_s$ is one sensing time.
For Fig.  \ref{fig:SCS_Nd30_ts1_rho0}, $t_s = 1$ ms and $N = 2048$ are used. The number of 
dwells $N_d$ is 30 and LPF bandwidth $B_f$ is set to 20 kHz to provide a buffer for
frequency deviations. The sensing requirements of 802.22 are achieved with these parameters.
Downsampling frequency $F_s$ and bandwidth $B_f$ are fixed throughout the rest of the paper.

\subsection{SCS Sensing Results and Comparisons}
\label{subsec:SimResults}

Before we evaluate the probability of detection, the analytical value of decision threshold is verified in Table \ref{tbl:DecThr} with simulation.
The decision threshold $\gamma$ is roughly proportional to $N_d$ as expected. It is slightly reduced as $t_s$ increases, but the effect is negligible compared to that of $N_d$. Also note that the analytical value fits better as the number of samples increases as expected.

The detection performance of the proposed SCS detector is compared with the FFT based pilot location detector \cite{CorGho07}, covariance absolute value (CAV) detection \cite{ZenLia09}, and CAV with low pass filtering and downsampling (step 2 of the SCS algorithm) in Fig. \ref{fig:cleanATSC} on detecting ideal ATSC signal with additive white Gaussian noise (AWGN). Sensing times for SCS and the pilot location detector are set to 30 ms (1 ms $\times$ 30 and 5 ms $\times$ 6 each) and CAV used  $N_s = 5\times10^5 $ samples (23.2ms) with a smoothing factor $L = 14$.
The decision thresholds for SCS and CAV are set to make $P_{FA} = 0.1$. The $P_{FA}$ of the pilot location detector was 0.02.
The sensitivities of SCS and pilot location detector are about 7 dB better than that of the CAV, while CAV with step 2
outperformed SCS about 1.3 dB. 
Note that the step 2 not only reduces the complexity but also significantly improves the sensitivity of the CAV detector due to the increased effective SNR. 
This step is also used in \cite{CorGho07}.

Though Fig. \ref{fig:cleanATSC} shows that SCS, pilot detector and improved CAV meet the sensitivity requirement of the IEEE 802.22 standard, it does not guarantee their performance in the real environment.
Since detector performance is highly dependent on the received signal quality that is degraded by the unpredictable wireless channel, it is very important to test the detector under practical circumstances.
Hence, all the simulation results from now on are conducted using the actual signal sets explained in Section \ref{subsec:SimProcedure}.

Figure \ref{fig:SCS_Nd_ts1_rho0} shows the superior sensitivity of SCS based on the number of dwells $N_d$.
All the parameters for each detector are set as used in Fig. \ref{fig:cleanATSC}.
The sensitivities are averaged over all 12 signals for each detection method.%
 \footnote{We corrected the average calculation in \cite{CorGho07}: it was done with 9 good signals.}
The SCS algorithm shows similar performance with pilot location sensing when $N_d$ = 6, which is 1/5th of the sensing time.
Fast sensing has several advantages. It can reduce interference to the primary service by enabling fast channel changes \cite{YucArs09}.
Furthermore, the secondary system can use the channel more efficiently since it cannot transmit during the sensing time \cite{CorCha06}.
When the same sensing time (30 ms) is used, it shows about 3 dB better sensitivity.
Figure \ref{fig:SCS_Nd_ts1_rho0} also shows that SCS is very effective detector in real environment. Though improved CAV (with low pass filtering and downsampling) has better sensitivity for ideal ATSC signal than SCS, the SCS outperforms it when detecting actual signals.

Given the sampling rate, the frequency resolution or equivalently number of FFT size $N$ is determined
by the duration of the sensing window $t_s$ as in (\ref{eq:N}). For example, 0.1 ms sensing window allows
a 128-point FFT, while 1 ms allows a 2048-point FFT. In Fig. \ref{fig:SCS_ts_Nd30_rho0}, we show
that fine frequency resolution (a longer sensing window) detects better when the same dwells are
used, which also coincides with the analytical results.
The sensitivity of the SCS algorithm is improved by roughly 2 dB per doubled window size.

Combining the above results, the relationship between the total sensing time $T_s = t_s \times N_d$ and
detection performance of SCS is investigated in Fig. \ref{fig:SCS_tsxNd_rho0}.
The simulation results show that $t_s$ affects performance more than $N_d$,
as expected from Theorem \ref{thm:P_D}.
Note that 12 ms sensing time with $t_s = 2$ ms has sensitivity of -19.9 dB, while 15 ms with $t_s = 0.5$ ms has -19.3 dB.
It clearly shows that there is a tradeoff between performance and complexity.
Increased sensing time allows an increased FFT size.
One way to reduce complexity is to decrease the down sampling frequency
$F_s$ and to use coarse frequency resolution. However, since a 2048-point FFT is included in every 802.22 device,
it can be reused in spectrum sensing. Thus, using a 1 ms sensing window will be a reasonable choice.
Sensitivity of -22.9 dB is
achieved with 2 ms $\times$ 30 = 60 ms sensing time.


The noise power has thus far been assumed to be known at the receiver. However, perfect
estimation of noise power is impossible in practice and it is well known that even
a fraction of dB noise uncertainty can degrade detector performance severely \cite{TanSah08}.
Figure \ref{fig:SCS_rho0vs2} shows the robustness of SCS to the noise uncertainty problem.
The noise uncertainty factor $\rho$ means that the noise
power has a uniform distribution of $[-\rho \,\mathrm{dB}, \rho \,\mathrm{dB}]$.  SCS shows almost
no performance degradation with 2dB of noise uncertainty regardless of sensing time.
Note that the nominal noise uncertainty level for the IEEE 802.22 system is 1 dB \cite{802.22_Req}.

The robustness of SCS mainly comes from the fact that it exploits the statistical independence of
the signal and noise components, especially that the noise is uncorrelated. Thus, uncertainty in the noise
only reduces the possible correlation and does not strongly affect signal detection.

\section{Conclusion}
\label{sec:conclusion}

In this paper, we have proposed a spectrum sensing algorithm using 
the covariance of the partial spectrum of the received signal.
The proposed SCS algorithm can be used to
detect arbitrary signals with a proper selection of parameters, depending on the features of the primary signal.
The decision threshold in particular should be carefully chosen since there is a fundamental tradeoff between the probability of false alarm and the probability of detection.
We derived an appropriate decision threshold and 
the detection time and sensitivity of SCS mathematically.
The theoretical results were verified through extensive simulations
conducted using actual digital TV signals captured in the US, and our code is openly posted online \cite{SCS-CODE-URL}.
%
The detection performance of SCS is compared with the FFT based pilot location detector and
the CAV detector. The SCS detector achieves the same detection performance as the FFT based pilot location detector in 1/5th of the sensing time
and 3 dB better sensitivity in the same sensing time of 30 ms. It has also been shown that SCS is highly robust
against noise uncertainty. 
The results of this paper suggest that SCS is an extremely effective sensing algorithm
and should be seriously considered by the 802.22 standards body, and future industry efforts at cognitive radio.

\appendices
\section{Proof of Lemma \ref{lem:mu_tau}}
\label{apdx:proof_mu_tau}

The spectrum can be categorized into three sections: noise only (lower frequency, $k_0$), pilot tone ($k_1$), and signal data part (high frequency, $k_2$).
Without loss of generality, we can set
$K_1$ at DC, i.e. $K_0 = \{-K, -K+1, \cdots, -1\}, \; K_1 = \{0\}, \mathrm{and } \; K_2 = \{1, \cdots, K\}$.

Then,
\begin{eqnarray}
  \mathbb{E} \left[\delta_w(\tau, k)\right] &=& \mathbb{E} \left[\delta_s(\tau, k)\right] = 0, \label{eq:Edeltaw} \\
  \mathbb{E} \left[m_\tau(k)\right] &=& 
	N_w + \delta_p \mathbf{1}_{K_1}(k) + N_s \mathbf{1}_{K_2}(k) \\
    \mathrm{var}\left(m_\tau(k)\right) &=& 
	\left( N_w + \delta_p \mathbf{1}_{K_1}(k) + N_s \mathbf{1}_{K_2}(k) \right)^2 
\end{eqnarray}
where $\delta_p = \mathbb{E} [\delta_p(\tau)]$ is the pilot power spectrum, $N_s = {\mathbb{E} [N_s(\tau)]}$ is the nominal PSD of the primary signal and $\mathbf{1}_A(k)$ is an indicator function.  The in-band average spectrum is given as
\begin{eqnarray}
    \mu_\tau &=& \frac{1}{|k|} \left\{ \sum_{k \in K_0} \left(N_w + \delta_w(\tau, k)\right) +
                \sum_{k \in K_1} \left(N_w + \delta_w(\tau,k) +  \delta_p(\tau)\right) + \right. \nonumber \\
                & & \;\;\;\;\;\;\; \left. \sum_{k \in K_2} \left(N_w + \delta_w(\tau, k) + N_s(\tau) + \delta_s(\tau,k)\right) \right\} \nonumber \\
		&=& N_w + \frac{\delta_p(\tau)}{2K} + \frac{N_s(\tau)}{2}, \label{eq:mu_tau_pf}
\end{eqnarray}
where
$ \delta_p = \frac{f_p P_S}{\Delta f} = \frac{f_p}{F_s/N}P_S$, $N_s = \frac{(1-f_p)P_S}{B}$
and $\Delta f = 1/t_s = F_s/N$ is the frequency resolution.

\section{Proof of Lemma \ref{lem:Ectu}}
\label{apdx:proof_Ectu}


  The expected value of $c_{\tau u}$ can be easily shown using (\ref{eq:m_tau_signal}), (\ref{eq:Nw+Ns}) and (\ref{eq:Edeltaw}) -- (\ref{eq:mu_tau_pf}).
  When there is signal, the first term of (\ref{eq:E[c_tau_u]}) can be analyzed in three frequency regions as 
  \begin{enumerate}
    \item $k \in K_0$: ~
      \begin{eqnarray}
	\mathbb{E}[\mathbf{m_\tau}^T \mathbf{m_u}] &=& \mathbb{E} \left[ \left( N_w(\tau) + \delta_w(\tau, k) \right)
			\left( N_w(u)+ \delta_w(u, k) \right) \right] \nonumber \\
	&=& N_w^2 + N_w^2 \alpha_w^{(\tau - u)} = N_w^2 \left(1 + \alpha_w^l \right),
	\label{eq:Emt_k0}
      \end{eqnarray}

    \item $k \in K_1$:
      \begin{eqnarray}
	\mathbb{E}[\mathbf{m_\tau}^T \mathbf{m_u}] &=& \mathbb{E} \left[ \left( N_w(\tau) + \delta_w(\tau, k) + \delta_p(\tau)\right)
			\left( N_w(u)+ \delta_w(u, k) + \delta_p(u) \right) \right] \nonumber \\
	&=& N_w^2 \left( 1 + \alpha_w^l \right) + 2N_w \delta_p + \alpha_p^l \delta_p^2,
	\label{eq:Emt_k1}
      \end{eqnarray}

    \item $k \in K_2$:
      \begin{eqnarray}
	\mathbb{E}[\mathbf{m_\tau}^T \mathbf{m_u}] &=& \mathbb{E} \big[ \left( N_w(\tau) + \delta_w(\tau, k) + N_s(\tau) + \delta_s(\tau,k) + \frac{2}{N}\mathrm{Re}\{S_\tau(k)W_\tau^\ast(k)\}\right) \cdot \nonumber \\
			&~& \quad \quad \left( N_w(u)+ \delta_w(u, k) + N_s(u) + \delta_s(u,k) + \frac{2}{N}\mathrm{Re}\{S_u(k)W_u^\ast(k)\} \right) \big] \nonumber \\
			&=& \mathbb{E}[N_w(\tau)N_w(u)] + \mathbb{E}[\delta_w(\tau,k) \delta_w(u,k)] + 2N_w N_s + \mathbb{E}[N_s(\tau) N_s(u)] \nonumber \\
			&~& \;\;\;\; + \mathbb{E}[\delta_s(\tau) \delta_s(u)] \label{eq:Emt_k2_1} \\
	&=& N_w^2 \left( 1 + \alpha_w^l \right) + 2N_w N_s + 2N_s(N_s + N_w)\alpha_s^l,
	\label{eq:Emt_k2}
      \end{eqnarray}
      \begin{equation*}
	\begin{array}{ll}
	  \because & %
	\mathbb{E}[\mathrm{Re}\{S_\tau(k)W_\tau^\ast (k)\}] %
	= \mathbb{E}[\mathrm{Re}\{S_\tau(k)\} \mathrm{Re}\{W_\tau(k)\}] + \mathbb{E}[\mathrm{Im}\{S_\tau(k)\} \mathrm{Im}\{W_\tau(k)\}] = 0, \\
	&  \mathbb{E}[\delta_s^2(\tau, k)] = \mathrm{var}\left( m_\tau(k) \right) - \mathbb{E}[\delta_w^2(\tau, k)] %
	= N_s^2 + 2N_w N_s.
      \end{array}
      \end{equation*}
  \end{enumerate}

  The expected value of the second term of (\ref{eq:E[c_tau_u]}) can be evaluated straight forwardly:
  \begin{eqnarray}
    \mathbb{E}[\mu_\tau \mu_u] &=& \mathbb{E}\left[ \left( N_w + \frac{\delta_p (\tau) }{2K} \right) \left( N_w + \frac{\delta_p(u)}{2K} \right) \right] \nonumber \\
	  &=& N_w^2 + \frac{N_w \delta_p}{K} + \frac{\alpha_p^l}{4K^2} \delta_p^2 \label{eq:Emumu}.
  \end{eqnarray}

  Now the expectation of the sample covariance is given as
  \begin{eqnarray}
    \mathbb{E}[c_{\tau u}] &=& \frac{2K+1}{2K} N_w^2 \left( 1 + \alpha_w^l \right) + \frac{\delta_p}{2K+1} (2N_w + \alpha_p^l \delta_p)
      + N_s \left( N_w + \left(N_s+N_w\right) \alpha_s^l \right) \nonumber \\
      & & - \frac{2K+1}{2K} \left( N_w^2 + \frac{N_w\delta_p}{K} + \frac{\alpha_p^l}{4K^2}\delta_p^2 \right) \nonumber \\
      &=& \alpha_w^l N_w^2 + \frac{\alpha_p^l}{2K}\delta_p^2 + N_s \left(N_w + (N_s + N_w)\alpha_s^l \right) + O(\frac{1}{K^2}). \label{eq:Ectu_sigPf} 
  \end{eqnarray}


  Next, we will derive the variance of $c_{\tau u}$ for the signal case. 
  Let
  \begin{equation}
    V_{\tau}^m (k) = m_\tau(k) - \mu_\tau.
    \label{eq:V_tu_k}
  \end{equation}
  Then the variance of $c_{\tau u}$ can be obtained as
  \begin{equation}
    \mathrm{var}(c_{\tau u}) = \mathrm{var}\left(\frac{1}{2K}\sum_k V_{\tau}^m(k) V_{u}^m(k) \right) 
    = \frac{1}{4K^2} \sum_k \mathrm{var}\left(V_{\tau}^m(k) V_{u}^m(k) \right),
    \label{eq:Vctu_def}
  \end{equation}
  where $V_{\tau}^m(k)$ and $V_{u}^m(k)$ are i.i.d and $V_{\tau}^m(k) V_{u}^m(k)$ and $V_{\tau}^m(j) V_{u}^m(j)$ are uncorrelated if $k \neq j$.

  Since the variance of the product of two random variables $X$ and $Y$ is given by
  \begin{equation}
    \mathrm{var}(XY) = (\mathbb{E}[X])^2 \sigma_X^2 + (\mathbb{E}[Y])^2 \sigma_Y^2 + 2\mathbb{E}[X]\mathbb{E}[Y]\mathrm{cov}(X,Y) + \sigma_X^2 + \sigma_Y^2 + \left(\mathrm{cov}(X,Y)\right)^2,
    \label{eq:varXY}
  \end{equation}
  where $\sigma_X^2$ and $\sigma_Y^2$ are the variances of $X$ and $Y$ respectively. We need to find the statistics of $V_{\tau}^m(k)$ and $V_{u}^m(k)$ first. It can be shown that
  \begin{equation}
    \mathbb{E}[V_{\tau}^m(k)] = \mathbb{E}[V_{u}^m(k)] = \begin{cases}
	    -\frac{\delta_p}{2K}, & k \in K_0,\\
	\delta_p, & k \in K_1, \\
	N_s - \frac{\delta_p}{2K}, & k \in K_2,
      \end{cases} 
  \end{equation}
  \begin{equation}
      \mathrm{var}(V_\tau^m(k)) = \mathrm{var}(V_u^m(k)) = \mathrm{var}(m_\tau(k)) 
      = (N_w + \delta_p \mathbf{1}_{K_1}(k) + N_s \mathbf{1}_{K_2}(k))^2,
  \end{equation}
  \begin{equation}
      \mathrm{cov}\left(V_\tau^m(k), V_u^m(k) \right) = \mathrm{var}(V_\tau^m(k)) \delta(\tau - u).
    \label{eq:EVtmk}
  \end{equation}

  Now we can find the variances in (\ref{eq:Vctu_def}).
  Define $C_j^k \triangleq \mathrm{var}(\left(V_{\tau}^m(k) V_{u}^m(k) \right), k \in k_j$, then
      \begin{eqnarray}
	C_0^k &=& \mathbb{E}\left[-\frac{\delta_p(\tau)}{2K}\right]^2 N_w^2 +  \mathbb{E}\left[-\frac{\delta_p(u)}{2K})\right]^2 N_w^2 \nonumber \\
	& & + 2\mathbb{E}\left[-\frac{\delta_p(\tau)}{2K}\right] \mathbb{E}\left[-\frac{\delta_p(u)}{2K}\right] N_w^2 \delta(\tau - u) + N_w^2 + N_w^2 + N_w^4\delta(\tau-u) \nonumber \\
	  &=& N_w^2 \left( \frac{\delta_p^2}{2K^2} \left( 1 + \alpha_p^l \delta(\tau-u) \right) + 2 + N_w^2 \delta(\tau - u) \right) + O(\frac{1}{K^4}),%
	\label{eq:C_0} \\
	C_1^k &=&  2\left( N_w + \delta_p^2\right)^2 %
	    \left( \delta_p^2 \left(1 + \alpha_p^l \delta(\tau-u) \right) + (N_w + \delta_p)^2 \delta(\tau-u) \right) + O(\frac{1}{K^4}),%
	\label{eq:C_1} \\
	      C_2^k &=&  2N_w^2 \left( \frac{\delta_p^2}{4K^2} \left( 1 + \alpha_p^l \delta(\tau - u) \right) + %
	       (N_w + N_s)^2 \delta(\tau-u) \right) + O(\frac{1}{K^4}).%
	\label{eq:C_2}
      \end{eqnarray}


  Thus, the variance of $c_{\tau u}$ can be obtained from (\ref{eq:Vctu_def}) as
  \begin{equation}
    \mathrm{var}(c_{\tau u}) = \frac{1}{4K^2} \left( KC_0^k + C_1^k + KC_2^k \right),
  \end{equation}
  when there is a detected signal. The no signal case $\mathcal{H}_0$ 
  is obtained by letting $N_s = 0$.



\bibliographystyle{IEEEtran}
\bibliography{JKim}

\begin{thebibliography}{10}
\providecommand{\url}[1]{#1}
\csname url@samestyle\endcsname
\providecommand{\newblock}{\relax}
\providecommand{\bibinfo}[2]{#2}
\providecommand{\BIBentrySTDinterwordspacing}{\spaceskip=0pt\relax}
\providecommand{\BIBentryALTinterwordstretchfactor}{4}
\providecommand{\BIBentryALTinterwordspacing}{\spaceskip=\fontdimen2\font plus
\BIBentryALTinterwordstretchfactor\fontdimen3\font minus
  \fontdimen4\font\relax}
\providecommand{\BIBforeignlanguage}[2]{{%
\expandafter\ifx\csname l@#1\endcsname\relax
\typeout{** WARNING: IEEEtran.bst: No hyphenation pattern has been}%
\typeout{** loaded for the language `#1'. Using the pattern for}%
\typeout{** the default language instead.}%
\else
\language=\csname l@#1\endcsname
\fi
#2}}
\providecommand{\BIBdecl}{\relax}
\BIBdecl

\bibitem{CabMis04}
D.~Cabric, S.~M. Mishra, and R.~W. Brodersen, ``Implementation issues in
  spectrum sensing for cognitive radios,'' in \emph{Asilomar Conference on
  Signals, Systems and Computers}, vol.~1, Nov. 2004, pp. 772--776.

\bibitem{CorGho07}
C.~Cordeiro, M.~Ghosh, D.~Cavalcanti, and K.~Challapali, ``Spectrum sensing for
  dynamic spectrum access of {TV} bands,'' in \emph{Intl. Conf. on Cognitive
  Radio Oriented Wireless Networks and Communications}, Orlando, FL, USA, Aug.
  2007, pp. 225--233.

\bibitem{BroWol04}
R.~Broderson, A.~Wolisz, D.~Cabric, S.~Mishra, and D.~Willkomm, ``{White paper:
  {CORVUS}: A cognitive radio approach for usage of virtual unlicensed
  spectrum},'' \emph{University of California, Berkeley, Tech. Rep}, 2004.

\bibitem{Hoven:thesis}
N.~K. Hoven, ``On the feasibility of cognitive radio,'' Master's thesis, Univ.
  California Berkeley, 2005.

\bibitem{GeiTon07}
S.~Geirhofer, L.~Tong, and B.~M. Sadler, ``Cognitive radios for dynamic
  spectrum access - dynamic spectrum access in the time domain: Modeling and
  exploiting white space,'' \emph{{IEEE} Comm. Mag.}, vol.~45, no.~5, pp.
  66--72, May 2007.

\bibitem{ZenLia09}
Y.~Zeng and Y.-C. Liang, ``Spectrum-sensing algorithms for cognitive radio
  based on statistical covariances,'' \emph{IEEE Transactions on Vehicular
  Technology}, vol.~58, no.~4, pp. 1804--1815, May 2009.

\bibitem{FCC08-260}
\BIBentryALTinterwordspacing
{Federal Communications Commission}, ``{FCC} 08-260: Second report and order
  and memorandum opinion and order,'' Nov. 2008. [Online]. Available:
  \url{www.fcc.gov}
\BIBentrySTDinterwordspacing

\bibitem{802.22/D1.0_2008}
{IEEE 802.22/D1.0}, ``{Draft Standard for Wireless Regional Area Networks Part
  22: Cognitive Wireless RAN Medium Access Control (MAC) specifications:
  Policies and procedures for operation in the TV bands},'' Apr. 2008.

\bibitem{SteCho09}
C.~R. Stevenson, G.~Chouinard, Z.~Lei, W.~Hu, S.~J. Shellhammer, and
  W.~Caldwell, ``{IEEE} 802.22: The first cognitive radio wireless regional
  area network standard,'' \emph{{IEEE} Communications Magazine}, vol.~47,
  no.~1, pp. 130--138, Jan. 2009.

\bibitem{Shellhammer08}
S.~J. Shellhammer, ``Spectrum sensing in {IEEE} 802.22,'' in \emph{IAPR Wksp.
  Cognitive Info. Processing}, Santorini, Greece, Jun. 2008.

\bibitem{YucArs09}
T.~Yucek and H.~Arslan, ``A survey of spectrum sensing algorithms for cognitive
  radio applications,'' \emph{IEEE Comm. Surveys \& Tutorials}, vol.~11, no.~1,
  pp. 116--130, First Quarter 2009.

\bibitem{Urkowitz67}
H.~Urkowitz, ``Energy detection of unknown deterministic signals,''
  \emph{Proceedings of the {IEEE}}, vol.~55, no.~4, pp. 523--531, Apr. 1967.

\bibitem{SonFis92}
A.~Sonnenschein and P.~M. Fishman, ``Radiometric detection of spread-spectrum
  signals in noise of uncertain power,'' \emph{{IEEE} Transactions on Aerospace
  and Electronic Systems}, vol.~28, no.~3, pp. 654--660, Jul. 1992.

\bibitem{TanSah08}
R.~Tandra and A.~Sahai, ``{SNR} walls for signal detection,'' \emph{IEEE
  Journal of Sel. Topics in Signal Processing}, vol.~2, no.~1, pp. 4--17, Feb.
  2008.

\bibitem{CheGao07}
H.-S. Chen, W.~Gao, and D.~G. Daut, ``Signature based spectrum sensing
  algorithms for {IEEE} 802.22 {WRAN},'' in \emph{{IEEE} International
  Conference on Communications (ICC)}, Glasgow, Jun. 2007, pp. 6487--6492.

\bibitem{HanSho06}
N.~Han, S.~Shon, J.~H. Chung, and J.~M. Kim, ``Spectral correlation based
  signal detection method for spectrum sensing in ieee 802.22 wran systems,''
  in \emph{Intl. Conf. Advanced Communication Technology}, vol.~3, Korea, Feb.
  2006, pp. 1765 --1770.

\bibitem{TanSah07}
R.~Tandra and A.~Sahai, ``{SNR} walls for feature detectors,'' in \emph{IEEE
  Intl. Symp. on New Frontiers in Dynamic spectrum Access Networks (DySPAN)},
  Dublin, Ireland, Apr. 17--20, 2007, pp. 559--570.

\bibitem{Tandra:thesis}
R.~Tandra, ``Fundamental limites of detection in low {SNR},'' Master's thesis,
  University of California Berkeley, 2005.

\bibitem{SCS-CODE-URL}
\BIBentryALTinterwordspacing
J.~Kim, ``Simulation code for spectral covariance sensing ({SCS}).'' [Online].
  Available:
  \url{http://www.ece.utexas.edu/$\thicksim$jaeweon/spectral\_covariance\_sens%
ing.html}
\BIBentrySTDinterwordspacing

\bibitem{VanVet09}
P.~Vandewalle, J.~Kovacevic, and M.~Vetterli, ``Reproducible research in signal
  processing - what, why, and how,'' \emph{IEEE Sig. Proc. Mag.}, vol.~26,
  no.~3, pp. 37--47, May 2009.

\bibitem{ProakisBook96}
J.~G. Proakis and D.~G. Manolakis, \emph{Digital Signal Processing},
  3rd~ed.\hskip 1em plus 0.5em minus 0.4em\relax Prentice-Hall, 1996.

\bibitem{802.22_Req}
S.~Shellhammer and G.~Chouinard, ``{Spectrum Sensing Requirements Summary},''
  IEEE 802.22-06/0089r1, Jun. 2006.

\bibitem{ATSC:A74}
{Advanced Television Systems Committee}, ``{ATSC Recommended Practice: Receiver
  Performance Guidelines (with Corrigendum No. 1 and Amendment No. 1)},'' ATSC
  A/74, Nov. 2007.

\bibitem{802.22_Sim_model}
S.~Shellhammer, V.~Tawil, G.~Chouinard, M.~Muterspaugh, and M.~Ghosh,
  ``{Spectrum Sensing Simulation Model},'' IEEE 802.22-06/0028r10, Sep. 2006.

\bibitem{802.22:InitialSigProc}
M.~Mathur, R.~Tandra, S.~Shellhammer, and M.~Ghosh, ``{Initial signal
  processing of captured DTV signals for evaluation of detection algorithms},''
  IEEE 802.22-06/0158r4, Sep. 2006.

\bibitem{CorCha06}
C.~Cordeiro, K.~Challapali, and M.~Ghosh, ``Cognitive {PHY} and {MAC} layers
  for dynamic spectrum access and sharing of {TV} bands,'' in \emph{intl. wksp.
  Techn. and policy for accessing spectrum (TAPAS)}.\hskip 1em plus 0.5em minus
  0.4em\relax New York, NY, USA: ACM, 2006, p.~3.

\end{thebibliography}

\newpage

\begin{table}
  \centering
  \caption{Notations used in Analysis}
  \label{tbl:Notations}
  \begin{tabular}{|c|c|c|}
    \hline
	Variable & Description & Unit \\ \hline
	$s_\tau(n)$, $w_\tau(n)$ & signal and noise samples on the $\tau$-th dwell, $s_\tau(n) \triangleq s(n+\tau N)$ & V \\
	$m_\tau(k)$ & $k$-th component of the spectrogram vector on the $\tau$-th dwell $\mathbf{m}_\tau$ & W/Hz \\
	$N_s(\tau)$, $N_w(\tau)$ & average signal and noise spectrum computed in the signal band at the $\tau$-th dwell & W/Hz \\
	$\delta_s(\tau,k)$, $\delta_w(\tau,k)$ & residual signal and noise components at the $k$-th bin and the $\tau$-th dwell & W/Hz \\
	$\delta_p(\tau)$  & power of the pilot tone at the $\tau$-th dwell & W/Hz \\
	$S_\tau(k)$, $W_\tau(k)$ & fourier transforms of $s_\tau(n)$ and $w_\tau(n)$ & \\
	$\mu_\tau$ & in-band average spectrum at the $\tau$-th dwell & W/Hz \\
	$N_o = N_w$ & nominal one-sided PSD of the noise signal & W/Hz \\
	$N_s$ 	   & nominal signal PSD & W/Hz \\
	$\delta_p$ & nominal pilot power & W/Hz \\
	$c_{\tau u}$ & sample covariance of $\mathbf{m}_\tau$ and $\mathbf{m}_u$ & \\
	$\alpha_p^l$, $\alpha_s^l$ & nomalized correlations for pilot and signal each & \\
	$\alpha_w^l$ & nomalized correlation for noise term & \\
	$A_{N_d}^p$, $A_{N_d}^w$ & accumulated correlations of pilot and noise & \\ \hline
  \end{tabular}
\end{table}


\begin{table}
  \centering
  \caption{Decision Thresholds for ATSC signal detection with $P_{FA} = 0.1$}
  \label{tbl:DecThr}
  \begin{tabular}{|c|c|c|c|c|}
    \hline
    \multicolumn{5}{|c|}{Analytical Value (\ref{eq:gamma})} \\ \hline
    \multirow{2}{0.5cm}{$N_d$}& \multicolumn{4}{c|}{$t_s$ (ms)} \\ \cline{2-5}
    	& 0.1    & 0.5   & 1.0   & 2.0 \\ \hline
	6	&  3.39  & 2.83  & 2.73  & 2.70 \\ \hline
	12	&  5.92  & 4.89  & 4.70  & 4.64 \\ \hline
	30	&  13.16 & 10.96 & 10.60 & 10.36 \\ \hline \hline
    \multicolumn{5}{|c|}{Simulated Value}  \\ \hline
    \multirow{2}{0.5cm}{$N_d$}& \multicolumn{4}{c|}{$t_s$ (ms)} \\ \cline{2-5}
    	& 0.1    & 0.5   & 1.0   & 2.0 \\ \hline
     	6	& 3.59     & 3.07    & 2.92   & 2.75 \\ \hline
     	12	& 6.51     & 5.27    & 4.29   & 4.79 \\ \hline
    	30	& 14.17    & 11.27   & 10.95  & 10.62 \\ \hline
  \end{tabular}
\end{table}


\begin{figure}[p]
  \centering
  \psfig{figure=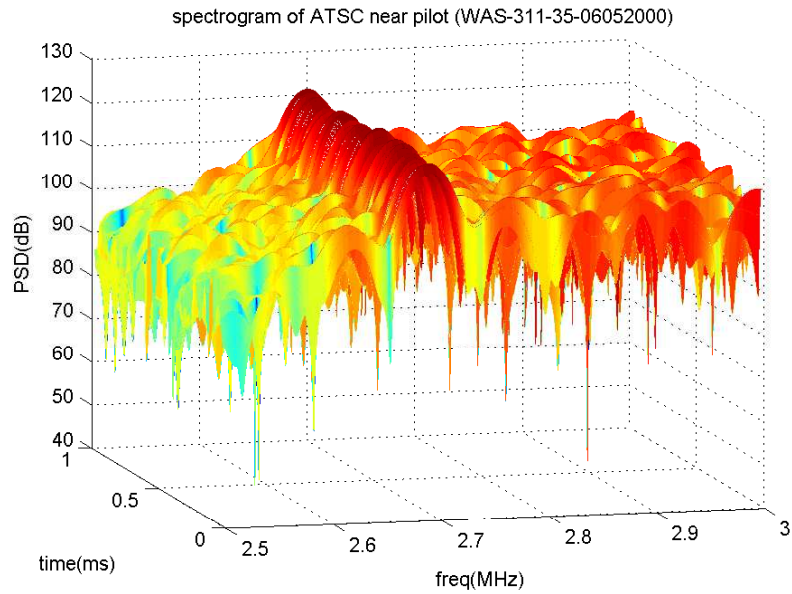,width=10.5cm}
  \caption{Spectrogram of a captured ATSC signal around pilot tone at 2.69 MHz.}
  \label{fig:ATSCspectrum}
\end{figure}

\begin{figure}[p]
  \centering
  \psfig{figure=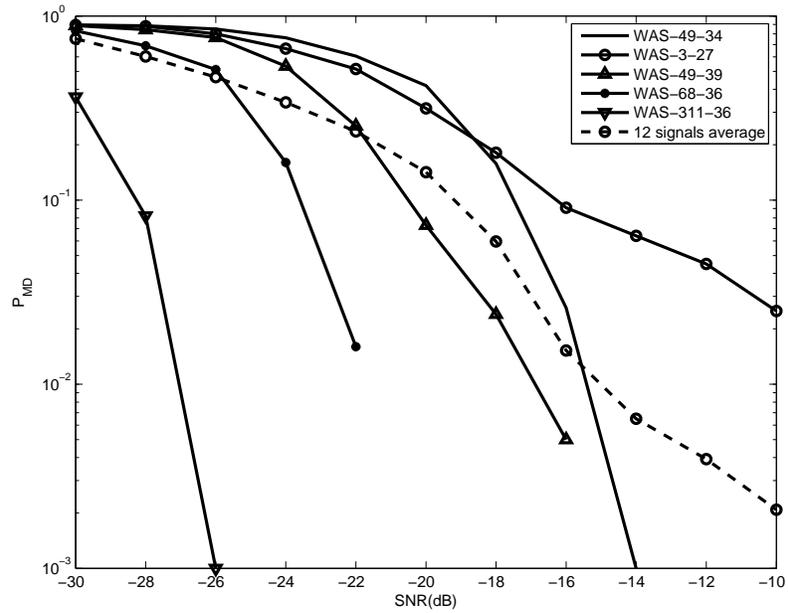, width=10.5cm}
  \caption{Probability of missed detection of Spectral Covariance Sensing (SCS) for the 5 ATSC signals selected from the 12 signals set, $t_s$ = 1ms, $N_d$ = 30.}
  \label{fig:SCS_Nd30_ts1_rho0}
\end{figure}

\begin{figure}[p]
  \centering
  \psfig{figure=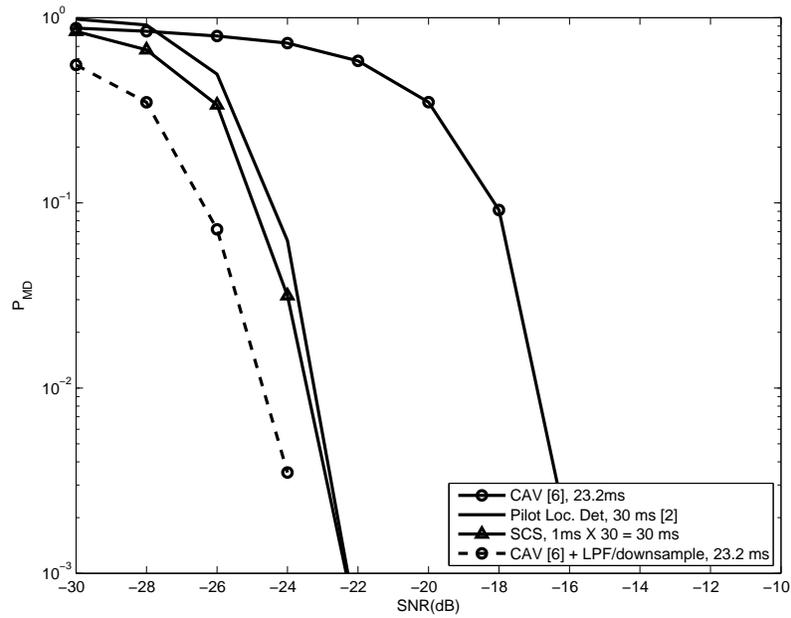, width=10.5cm}
  \caption{Probability of missed detection of Spectral Covariance Sensing (SCS), FFT pilot location detector \cite{CorGho07}, Covariance Absolute Value (CAV) \cite{ZenLia09}, and CAV with low pass filter and downsampling (step 2 of SCS algorithm) on the ideal ATSC signal.}
  \label{fig:cleanATSC}
\end{figure}

\begin{figure}[p]
  \centering
  \psfig{figure=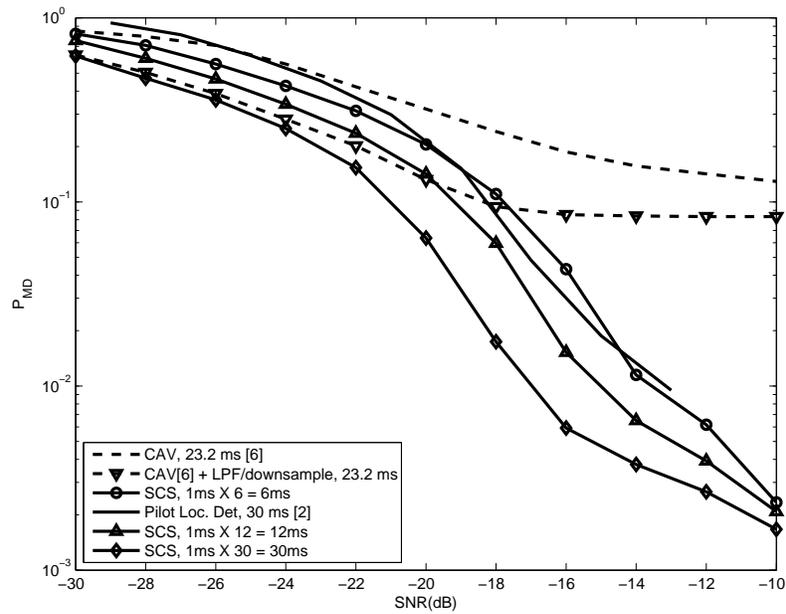, width=10.5cm}
  \caption{Effects of $N_d$ in Spectral Covariance Sensing ($t_s$ = 1ms) compared with the FFT pilot location detector \cite{CorGho07} and covariance absolute value (CAV) detector \cite{ZenLia09}.} 
  \label{fig:SCS_Nd_ts1_rho0}
\end{figure}

\begin{figure}[p]
  \centering
  \psfig{figure=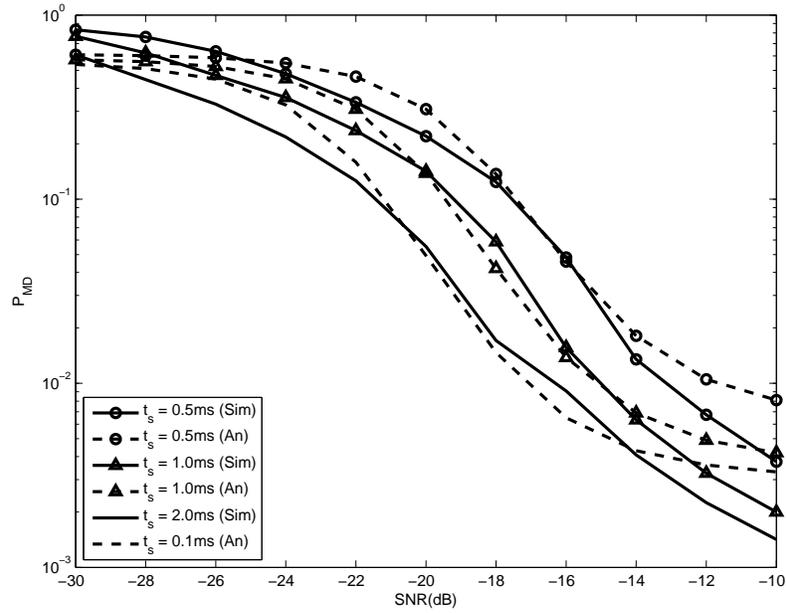, width=10.5cm}
  \caption{Effects of $t_s$ in Spectral Covariance Sensing, $N_d$ = 12.} 
  \label{fig:SCS_ts_Nd30_rho0}
\end{figure}

\begin{figure}[p]
  \centering
  \psfig{figure = 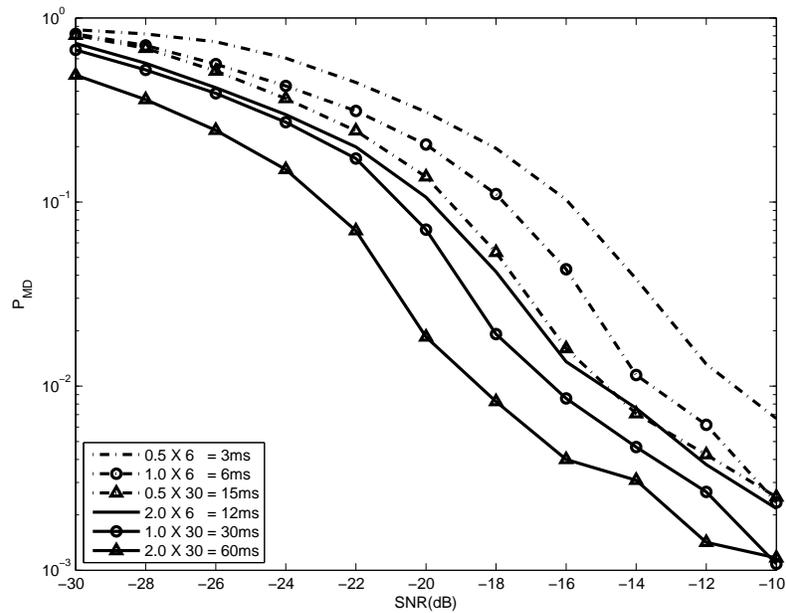, width = 10.5cm}
  \caption{Probability of missed detection of Spectral Covariance Sensing in various sensing times ($t_s \times N_d$).} 
  \label{fig:SCS_tsxNd_rho0}
\end{figure}

\begin{figure}[p]
  \centering
  \psfig{figure = 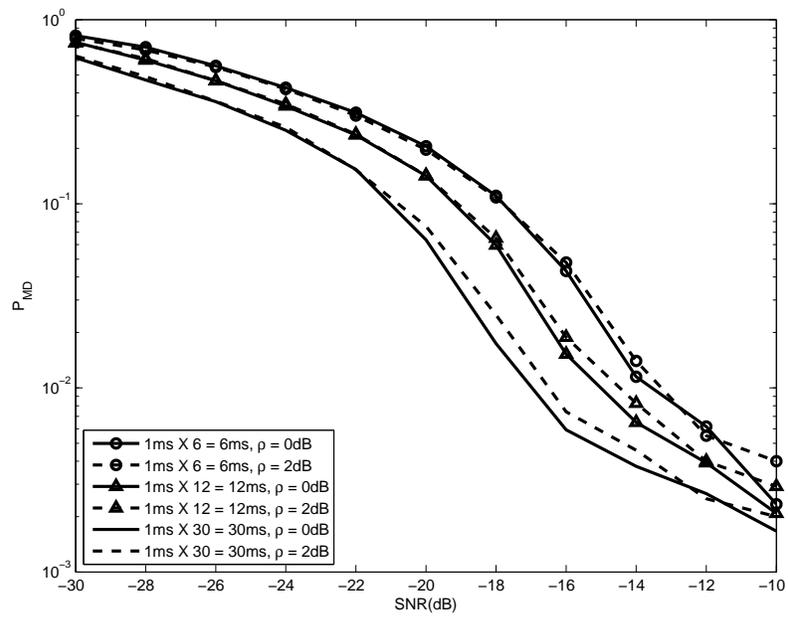, width = 10.5cm}
  \caption{Comparison of the probability of missed detection of Spectral Covariance Sensing with 2 dB noise uncertainty and 0 dB (without uncertainty).} 
  \label{fig:SCS_rho0vs2}
\end{figure}

\end{document}